\documentclass[%
 reprint,
 amsmath,amssymb,
 aps,
]{revtex4-1}

\usepackage{graphicx}
\usepackage{dcolumn}
\usepackage{bm}

\usepackage{amsmath}
\usepackage{amssymb}
\usepackage{amsthm}
\usepackage[utf8]{inputenc}
\usepackage[english]{babel}

\newtheorem{theorem}{Theorem}

\begin{document}

\preprint{APS/123-QED}

\title{Non-Gaussianity and entropy-bounded uncertainty relations: Application to detection of non-Gaussian entangled states}

\author{Kyunghyun Baek}
\affiliation{Department of Physics, Texas A$\&$M University at Qatar, Education City, P.O. Box 23874, Doha, Qatar}%
\author{Hyunchul Nha}%
\affiliation{Department of Physics, Texas A$\&$M University at Qatar, Education City, P.O. Box 23874, Doha, Qatar}%


\begin{abstract}
We suggest an improved version of Robertson-Schr\"odinger uncertainty relation for canonically conjugate variables by taking into account a pair of characteristics of states: non-Gaussianity and mixedness quantified by using fidelity and entropy, respectively. This relation is saturated by both Gaussian and Fock states, and provides strictly improved bound for any non-Gaussian states or mixed states. For the case of Gaussian states, it is reduced to the entropy-bounded uncertainty relation derived by Dodonov. Furthermore, we consider readily computable measures of both characteristics, and find weaker but more readily accessible bound. With its generalization to the case of two-mode states, we show applicability of the relation to detect entanglement of non-Gaussian states.
\end{abstract}

\maketitle

\section{Introduction}

Ever since the birth of quantum mechanics, the Heisenberg uncertainty principle \cite{Heisenberg1927} has played a key role to understand and explore fundamental nature of quantum mechanics. The first rigorously proven uncertainty relation for the case of canonically conjugate variables like the position and momentum satisfying $[\hat x, \hat p]=i$ is due to Kennard \cite{Kennard1927} and Weyl \cite{Weyl1928} and implies a fundamental limitation on preparing quantum states well-localized in both position $x$ and momentum $p$ space. As its generalization, Robertson-Schr\"odinger(RS) uncertainty relation was derived for a pair of Hermitian operators \cite{Robertson1929,Schrodinger1930}.

For the case of canonically conjugate variables, it is well-known that the RS relation provides an inequality invariant under symplectic transformations and saturated for pure Gaussian states. To improve its bound, on the one hand, additional information on mixedness was considered in the form of generalized purities\cite{Bastiaans1984,Dodonov2002} or von Neumann entropy \cite{Bastiaans1986,Dodonov2002}. More recently, on the other hand, with a better understanding of the notion of non-Gaussianity \cite{Genoni2008,Genoni2010} the RS relation was rederived under constraints on the degree of Gaussianity \cite{Mandilara2012} that was subsequently generalized to involve purity \cite{Mandilara2014}. Additionally, in the context of entropic uncertainty relations \cite{Hirschman1957,Biaynicki1975,Beckner1975} providing a stronger uncertainty bound than the Kennard-Weyl uncertainty relation, the role of non-Gaussianity was explicitly examined \cite{Son2015} and applied to generalize it in terms of entropy power \cite{Hertz2017}.

Beyond the fundamental interest with the improvement of uncertainty relations for non-Gaussian states, another important motivation is that improved uncertainty relations can be used to derive entanglement criteria for continuous variable systems \cite{Duan2000,Simon2000,Shchukin2005,Hillery2006,Nha2008,Walborn2009,Nha2012,Hertz2016}.
However, in spite of these efforts, the verification of entanglement for non-Gaussian states, which play an important role in quantum information processing \cite{Horodecki2009,Lloyd1999,Nha2004,*Garcia2004,Eisert2002,*Fiurasek2002,*Giedke2002,Niset2009}, has not been completely resolved  partly due to the nonoptimality of uncertainty relations for non-Gaussian states, while it has been done for Gaussian states. 

In this work, our main goal is to present improved version of RS uncertainty relation by taking into account additional characteristics of states: mixedness and non-Gaussianity together. For this purpose, we consider von Neumann entropy as a measure of the mixedness and suggest a non-Gaussianity measure based on fidelity between a state and its reference Gaussian state with the same covariance matrix. We show that this non-Gaussianity measure obeys required properties as a legitimate measure of non-Gaussianity. By using this measure, we find improved RS relation with an uncertainty bound monotonically increasing with respect to both quantities and generalize it to two-mode states. Additionally, we present a less improved but readily computable version of the RS uncertainty relation by using more easily computable measures of both entropy and non-Gaussianity. Finally, we show its applicability to detect entanglement of non-Gaussian states in combination with partial transposition.

\section{Variance-based uncertainty relations}

For the case of the canonically conjugate variables, the RS uncertainty relation can be written in terms of a covariance matrix, 
\begin{align}\label{RSUR}
\sqrt{\det {{V}}}\geq \frac{1}{2}.
\end{align}
Here, $V$ is a covariance matrix of $\hat \rho$ whose matrix elements are given by
$
V_{ij}= \frac{1}{2} \langle \{\hat r_i,\hat r_j\}\rangle-\langle \hat r_i\rangle\langle \hat r_j\rangle
$
for a vector of quadrature operators $\hat{\vec{r}}=(\hat x,\hat p)$ 
with the expectation value $\langle \hat O \rangle=\text{Tr}[ \hat O \hat \rho]$ and the anticommutator $\{\cdot,\cdot\}$. This form of uncertainty relation allows one to straightforwardly identify the invariance of the RS relation under linear canonical transformations that transform the covariance matrix into $SVS^T$ under the symplectic transformation $S\in Sp_{2,\mathbb{R}}$ with $\det S=1$ \cite{Simon1994}.

The set of states saturating the RS uncertainty relation consists of pure Gaussian states. That means any mixed or non-Gaussian state has a nontrivial covariance matrix whose $\sqrt{\det V}$ is bigger than 1/2. For mixed states, it was shown that the RS relation has larger bound with a fixed von Neumann entropy, 
\begin{align}\label{Entropy}
S(\rho)=-\text{Tr}[\hat\rho \ln \hat\rho],
\end{align}
in the form of
\begin{align}\label{EUR}
\sqrt{\det {V}}\geq h^{-1}\big(S(\rho)\big),
\end{align}
where the bound is given by the inverse of monotonically increasing function,
$h(x)=-\left(x-\frac{1}{2}\right) \ln\left(x-\frac{1}{2}\right)+\left(x+\frac{1}{2}\right) \ln\left(x+\frac{1}{2}\right)$, for $x>1/2$.
This relation is the so-called entropy-bounded uncertainty relation, and it is well known that it is saturated by all Gaussian states \cite{Dodonov2002}.

Also, generalized purities \cite{Bastiaans1984} can be considered as measures of the mixedness, and the so-called purity-bounded uncertainty relation was derived with a fixed purity in \cite{Dodonov2002}. However, only a specific set of non-Gaussian states is included in the states saturating the uncertainty relations. Therefore, we show here that one can appropriately improve the uncertainty relations by incorporating the entropy as a measure of mixedness together with a measure of non-Gaussianity to deal with a broad class of non-Gaussian states.

\section{Quantification of non-Gaussianity}

For an $N$-mode system described by mode operators $\hat a_k=(\hat x_k+ i \hat p_k)/\sqrt{2}$, a quantum state $\hat \rho$ is referred as a Gaussian state if its quasiprobability functions such as Wigner function are written in a Gaussian form, which in turn is fully determined by its first and second moments of $x$ and $p$. From an operational point of view, we can also refer to the Gaussian states as states generated by acting the linear canonical transformations on vacuum or thermal states, $\otimes_{k=1}^N \hat \tau(\bar n_k)$, where $\hat \tau(\bar n_k)=(1+\bar n_k)^{-1}(\bar n_k/(1+\bar n_k))^{\hat a^\dagger_k \hat a_k}$ is a thermal state of $k$th mode with average photon number $\bar n_k$. This transformation can be described by the symplectic transformation with the translation in the phase space of $x$ and $p$. 
For a single mode system, a Gaussian state can be generally expressed in the form of 
\begin{align}
\hat \rho_G= \hat D(\alpha) \hat S(\xi) \hat \tau(\bar n)  \hat S^\dagger(\xi) \hat D^\dagger(\alpha),
\end{align}
where $\hat D(\alpha)=\exp[\alpha \hat a^\dagger - \alpha^* \hat a]$ is the displacement operator and $\hat S(\xi)=\exp[(\xi (\hat a^\dagger)^2-\xi^* \hat a^2)/2]$ is the single-mode squeezing operator with $\alpha,\xi\in\mathbb C$.

The quantification of non-Gaussianity was proposed as measuring distance between a given state $\hat \rho$ and its reference Gaussian state $\hat \rho_G$ with the same first and second moments of $\hat\rho$ in terms of relative entropy \cite{Genoni2008,Genoni2010} and Hilber-Schmidt distance \cite{Genoni2007}. 
In a similar manner, we propose a non-Gaussianity measure by using Uhlmann fidelity as a measure of distance
\begin{align}
\mathcal{N}(\rho)=-2\ln F(\rho,\rho_G)=-2\ln \text{Tr}\left[\sqrt{\sqrt{\hat\rho}\hat\rho_G\sqrt{\hat\rho}}\right],
\end{align}
where the fidelity $F(\rho,\rho_G)$ characterizes a distance between a state $\hat\rho$ and its reference Gaussian state $\hat\rho_G$.
In accordance with the previous work \cite{Genoni2008}, appropriate properties as a measure of non-Gaussianity are examined as follows:

\begin{itemize}
	
	\item[($\mathcal N$1)] $\mathcal N(\rho)=0$ if and only if $\hat\rho$ is a Gaussian state, otherwise it gives rise to a nonzero positive value.
	
	{\it Proof.} Uhlmann fidelity between arbitrary states $\hat\rho$ and $\hat\sigma$ becomes unity if and only if they are equal, that is, $F(\rho,\sigma)=1$ if and only if $\hat\rho=\hat\sigma$. Thus, $\mathcal N(\rho)=0$ if and only if $\hat\rho=\hat\rho_G$, i.e. $\hat\rho$ is a Gaussian state. 
	
	\item[($\mathcal N$2)] $\mathcal N(\rho)$ is invariant under Gaussian unitary transformations. Namely, if $\hat U$ is a unitary operator corresponding to a symplectic transformation in the phase space, i.e. $\hat U=e^{-i \hat H}$ with Hamiltonian that is at most bilinear in the field operators, then $\mathcal N(U\rho U^\dagger)=\mathcal N(\rho)$.
	
	{\it Proof.}  After unitary transformation corresponding to a symplectic transformation, $\hat\rho$ is changed to $\hat U\hat\rho \hat U^\dagger$ and its reference Gaussian state $\hat\rho_G$ is also to $\hat U\hat\rho_G \hat U^\dagger$. Since the fidelity is invariant under the unitary transformation, i.e. $F(\rho,\rho_G)=F(U \rho U^\dagger, U\rho_G U^\dagger)$, $\mathcal N$ is invariant under Gaussian unitary transformations.
	
	\item[($\mathcal N$3)] $\mathcal N$ is additive for tensor products $\hat \rho^A\otimes \hat\rho^B$, i.e. ${\mathcal N(\rho^A\otimes\rho^B)=\mathcal N(\rho^A)+\mathcal N(\rho^B)}$.
	
	{\it Proof.}  Reference Gaussian state of $\hat\rho^A\otimes \hat\rho^B$ is given by $\hat\rho^A_G\otimes\hat\rho^B_G$. Thus, due to the multiplicavity of the fidelity, i.e. $F(\rho^A\otimes\rho^B,\sigma^A\otimes\sigma^B)=F(\rho^A,\sigma^A)F(\rho^B,\sigma^B)$, we have $\mathcal N(\rho^A\otimes \rho^B)=\mathcal N(\rho^A)+\mathcal N(\rho^B)$.
	
	\item[($\mathcal N$4)] $\mathcal N$ is non-increasing with respect to partial trace, $\mathcal N(\rho^{AB}) \geq \mathcal N(\rho^A)$.
	
	{\it Proof.} For a bipartite state $\hat\rho^{AB}$, its reduced state is defined as $\text{Tr}_B[\hat\rho^{AB}]=\hat\rho^A$. Covariance matrix of $\hat\rho^A$ is given by the same elements of  covariance matrix of $\hat\rho^{AB}$ only determined by the expectation values of $A$ system. Thus, reduced state of the reference Gaussian state of $\hat\rho^{AB}$ is equal to the reference Gaussian state of reduced state $\hat\rho^A$, that is $\text{Tr}_B[\hat\rho_G^{AB}]=\hat\rho_G^A$. Then, applying the monotonicity of the fidelity under partial trace, that is, $F(\rho^{AB},\rho_G^{AB})\leq F(\rho^A,\rho_G^A)$, one can prove ($\mathcal N$-4). 
	
	\item[($\mathcal N$5)] $\mathcal N$ monotonically decreases under Gaussian quantum channels, $\mathcal N(\rho) \geq \mathcal N(\mathcal E_G(\rho))$, where $\mathcal E_G$ is a Gaussian quantum channel.
	
	{\it Proof.}  Any Gaussian channel $\mathcal{E}_G$ can be written as $\mathcal{E}_G(\hat\rho)=\text{Tr}_E[\hat U(\hat\rho\otimes \hat\rho^E_G)\hat U^\dagger]$, where $\hat U$ corresponds to a symplectic transformation in phase space after including an ancillary system $E$ with its state $\hat\rho^E_G$. Then according to ($\mathcal N$-4), we have $\mathcal N\big(U(\rho\otimes\rho_G^E)U^\dagger\big)=\mathcal N\big(\rho\otimes\rho_G^E\big)$ and by taking partial trace over $E$ and using ($\mathcal N$-2) and ($\mathcal N$-3), we have $\mathcal N\big(\rho\big)\geq \mathcal N\big(\mathcal E_G(\rho)\big)$.
\end{itemize}

It is worth noting that our non-Gaussianity measure $\mathcal N$ is a special case of the quantum R\'enyi relative entropy \cite{Muller-Lennert2013} of $\hat \rho$ with respect to $\hat \sigma$, which is given by
\begin{align}\label{RenDiv}
S_\alpha(\rho\|\sigma)=\frac{1}{\alpha-1}\ln \left(\text{Tr}\left[\left(\hat\sigma^{\frac{1-\alpha}{2\alpha}}\hat\rho\hat\sigma^{\frac{1-\alpha}{2\alpha}}\right)^\alpha\right]\right),
\end{align}
for the order $\alpha\geq1/2$ and $\alpha\neq 1$. More specifically, $\mathcal N$ corresponds to the quantum R\'enyi relative entropy for $\alpha=1/2$ such that $S_{1/2}(\rho\|\sigma)=-2\ln F(\rho,\sigma)$. For $\alpha=1$ as another special case, it becomes the quantum relative entropy $S(\rho\|\sigma)=\text{Tr}[\hat \rho\ln \hat \rho-\hat\rho\ln\hat\sigma]$, which has been employed as a non-Gaussianity measure in \cite{Genoni2008}.

In general, the quantum R\'enyi relative entropy for any order $\alpha$ can be employed as a measure of non-Gaussianity since  $S_{\alpha}(\rho\|\rho_G)$ satisfies the required properties of the non-Gaussianity measure that can be shown by using its own properties addressed in \cite{Muller-Lennert2013}. However, as pointed out in \cite{Park2017}, these measures may not be readily computable because of the difficulty of solving the eigenvalue problem in the infinite dimensional Hilbert space. 
To deal with the difficulty, we consider the fidelity-based non-Gaussianity as one can alternatively adopt superfidelity \cite{Mendonca2008,Miszczak2009},
\begin{align}
G(\rho,\sigma)^2\equiv\text{Tr} [\hat \rho \hat \sigma]+\sqrt{1-\text{Tr} [\hat \rho^2]}\sqrt{1-\text{Tr} [\hat \sigma^2]},
\end{align}
that allows one to use the phse-space description for calculation. This quantity becomes unity if and only if $\hat \rho=\hat \sigma$, providing upper bound of the usual fidelity, i.e. $$F(\rho,\sigma)\leq G(\rho,\sigma).$$ 
That means alternatively one can employ it to quantify readily computable non-Gaussianity measure 
\begin{align}\label{Ng}
\mathcal{N}_g(\rho)= -2\ln G(\rho,\rho_G)\leq \mathcal{N}(\rho)
\end{align}
which gives strictly nonzero values for non-Gaussian states and also provides a lower bound for fidelity-based non-Gaussianity. 

\section{Non-Gaussianity-and entropy-bounded uncertainty relation}

In this section, we present an improved version of the RS relation by taking into account two characteristics of quantum states together: the non-Gaussianity and the mixedness quantified by the non-Gaussianity measure $\mathcal N(\rho)$ and the entropy $S(\rho)$ respectively. This inequality is referred to as non-Gaussianity-and entropy-bounded (NE) uncertainty relation.

\subsection{Non-Gaussianity-and entropy-bounded uncertainty relation for single-mode system}

For a single mode system, we obtain the NE uncertainty relation by incorporating the non-Gaussianity measure $\mathcal N(\rho)$ and the entropy $S(\rho)$ as follows:

\begin{theorem}
	For a single system $\hat \rho$ with the non-Gaussianity $\mathcal N(\rho)$ and the entropy $S(\rho)$, we have
	\begin{align}\label{NEUR}
	\sqrt{\det { V}}\geq h^{-1}(S(\rho)+\mathcal N(\rho)),
	\end{align}
	where $V$ is the covariance matrix of $\hat\rho$ and $h(x)=-\left(x-\frac{1}{2}\right) \ln\left(x-\frac{1}{2}\right)+\left(x+\frac{1}{2}\right) \ln\left(x+\frac{1}{2}\right)$ is a monotonically increasing function of $x>1/2$.
\end{theorem}
\begin{proof}
	According to the monotonicity of the quantum R\'enyi relative entropy with respect to the order $\alpha$, we have 
	\begin{align}
	S(\rho\|\rho_G)\geq S_{{1}/{2}}(\rho\|\rho_G). 
	\end{align}
	The left hand side can be rewritten as $S(\rho\|\rho_G)=\text{Tr}[\hat\rho\ln\hat\rho-\hat\rho\ln\hat\rho_G]=S(\rho_G)-S(\rho)$, since  $\ln\hat\rho_G$ is determined by at most bilinear in the field operators. Furthermore, the entropy of reference Gaussian state is explicitly expressed in terms of covariance matrix, $S(\rho_G)=h(\sqrt{\det V})$. Thus, by adding the entropy and taking inverse of $h(x)$ on both sides we have the NE uncertainty relation \eqref{NEUR}.
\end{proof}

As desired, the NE uncertainty relation is invariant under Gaussian unitary transformations, since $S(\rho)$, $\mathcal N(\rho)$ and $\det  V$ are all invariant. Furthermore, the NE uncertainty relation is saturated by all Gaussian states and a set of states provided by performing Gaussian unitary transformations on number states. More specifically, for the case of Gaussian states, the NE uncertainty relation reduces to the entropy-bounded uncertainty relation which is saturated for all Gaussian states. On the other hand,  for the case of pure states where the entropy vanishes, the inequality \eqref{NEUR} is also saturated by $|n\rangle$. Due to the invariance under Gaussian unitary transformation, we see that all states obtained by acting them on number states thus saturate the NE uncertainty relation.

Additionally, we can introduce a weaker but readily computable NE uncertainty relation,
\begin{align}\label{NEUR2}
\sqrt{\det { V}}\geq h^{-1}(-\ln\mu+\mathcal N_g(\rho)),
\end{align}
where $\mu=\text{Tr}[\hat \rho^2]$ is the purity. This relation is straightforwardly obtained by directly applying inequalities $S(\rho)\geq -\ln\mu$ and $\mathcal N(\rho)\geq \mathcal N_g(\rho)$ defined in Eq. \eqref{Ng}. All quantities in the bound can be readily calculated using phase-space distributions, and for mixed and non-Gaussian states, both quantities give us nontrivial positive values. 
{For the case of pure states, moreover, we note that the weaker inequality \eqref{NEUR2} becomes identical with the original one \eqref{NEUR} due to $S(\rho)=-\ln \mu=0$ and $\mathcal{N}(\rho)=\mathcal{N}_g(\rho)$.}

In what follows, let us investigate the tightness of inequality \eqref{NEUR} comparing with so-called 'putity-and Gaussianity-bounded(PG) uncertainty relation' \cite{Mandilara2012,Mandilara2014} to show its validity as a refined inequality. In the derivation of PG uncertainty relation, the degree of non-Gaussianity is characterized by
\begin{align}
g (\rho)=\frac{\text{Tr}\left[\hat\rho \hat\rho_G\right]}{\text{Tr}\left[\hat\rho_G^2\right]},
\end{align}
which is called Gaussianity. We note that $g$ holds the invariance under Gaussian unitary transformations and that it becomes unity, i.e. $ g(\rho)=1$, for Gaussian states. However, $ g(\rho)=1$ does not imply $\hat\rho$ is Gaussian, namely, $ g(\rho)=1$ is a necessary, but not sufficient condition for Gaussian states. 
{To focus on behaviors of NE and PG uncertainty relations with respect to the degree of non-Gaussianity, we consider non-Gaussian pure states in what follows. For this purpose, we specify the exact form of the PG uncertainty relation for the case of pure states \cite{Mandilara2012}, i.e. $\mu=1$,
\[\det V\geq
\left\{
\begin{aligned}
&\frac{g}{2(2-g)} && \text{for } g>1,\\
&\frac{2+2\sqrt{1-g}-g}{2g} && \text{for } \frac{2}{e}<g\leq1.
\end{aligned}
\right. \]
}

\begin{figure}[t]
	\centering	
	\includegraphics[scale =0.6]{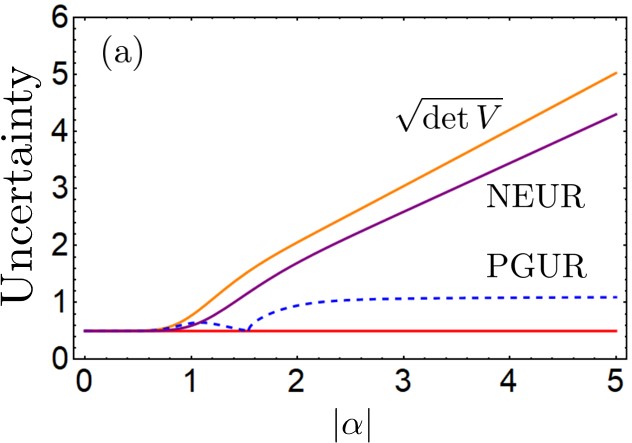}
		
	\includegraphics[scale =0.6]{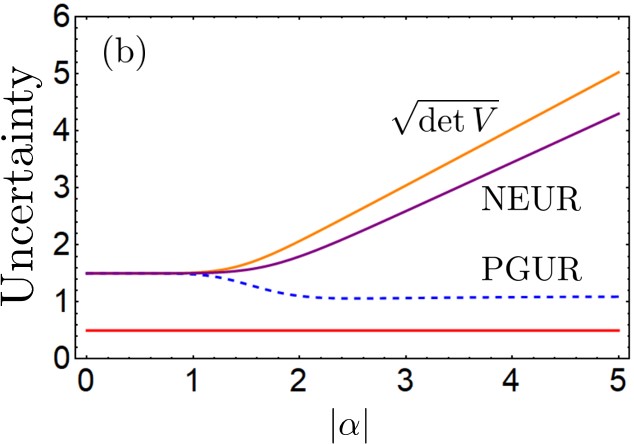} 
	
	\caption{Plots of $\sqrt{\det  V}$ and the bounds of the NE, PG and RS uncertainty relations for (a) even and (b) odd cat states against the  amplitude $|\alpha|$.}\label{CatUR}
\end{figure}

Now, as an example, let us consider even and odd cat states  \cite{Dodonov1974} defined by
\begin{align}
|\psi_{\pm}\rangle=\frac{1}{\sqrt{2(1\pm  e^{-2|\alpha|^2})}}(|\alpha\rangle\pm|-\alpha\rangle),
\end{align}
respectively, where $|\alpha\rangle=\hat D(\alpha)|0\rangle$ is a coherent state with $\alpha=|\alpha|e^{i\theta}$. Here, the angle $\theta$ represents the rotation angle so that it does not affect the degree of non-Gaussianity due to the invariance under the rotation operation. On the other hand, $|\alpha|$ represents the distance between two coherent states, $|\alpha\rangle$ and $|-\alpha\rangle$.
To illustrate our NE uncertainty relation as a refined one, we plot the NE and PG uncertainty relations for even and odd cat states according to the amplitude $|\alpha|$ in Figs. \ref{CatUR}(a) and \ref{CatUR}(b), respectively.

For the case of even cat states, they reduce to vacuum state as $|\alpha|$ goes to 0. Thus, when $|\alpha|$ is close to 0, all uncertainty relations provides optimal bound as observed in Fig. \ref{CatUR}(a), since they are all saturated for vacuum states. However, as $|\alpha|$ increases, uncertainty characterized by $\sqrt{\det  V}$ also increases. The bound of NE uncertainty relation shows the same behavior by taking into account non-Gaussianity effect, while the bound of PG uncertainty relation does not. while the PG uncertainty relation gives a stronger bound when $|\alpha|<1.1$, the NE uncertainty relation provides a tighter uncertainty bound in a larger region. In addition, the behavior that the bound of PG uncertainty relation becomes close to 1/2 at $|\alpha|\sim1.57$ is observed because the Gaussianity $g$ can have the value of unity for non-Gaussian states. On the other hand, for the case of odd cat states, they reduce to a single photon state as $|\alpha|$ goes to 0. In this limit, both bounds of NE and PG relations are saturated, since they are optimized for number states. However, as $|\alpha|$ increases,  $\sqrt{\det  V}$ and NE relation show similar behavior, while PG relation does not.\\

As another example of non-Gaussian pure state, let us consider photon-added coherent states provided by adding single photon to a coherent state \cite{Agarwal1991} described by
\begin{align}
|\psi_{pacs}\rangle= \frac{1}{\sqrt{1+|\alpha|^2}}\hat a^\dagger |\alpha\rangle.
\end{align}
In the limit $\alpha\rightarrow 0$, $|\psi_{pacs}\rangle$ reduces to $|1\rangle$, while in the another extreme, $\alpha\rightarrow \infty$, it tends to $|\alpha\rangle$. Thus, as $|\alpha|$ increases from 0, the effect of non-Gaussianity decreases. This trend is clearly observed in Fig. \ref{FigPACS}. 
When $|\alpha|$ is close to $0$, Fig. \ref{FigPACS} shows that both the NE and PG uncertainty relations give the optimal bound, since both are saturated for number states. However, as $|\alpha|$ increases, distinct behaviors are observed. Particularly, when $\alpha\sim 0.56$, $G(|\psi_{pacs})$ becomes unity but $|\psi_{pacs}\rangle$ is not Gaussian at this point. Further, in $|\alpha|\rightarrow\infty$, as $|\psi_{pacs}\rangle$ reduces to $|\alpha\rangle$, both relations are saturated again.

\begin{figure}[t]
	\centering
	\includegraphics[scale =0.3]{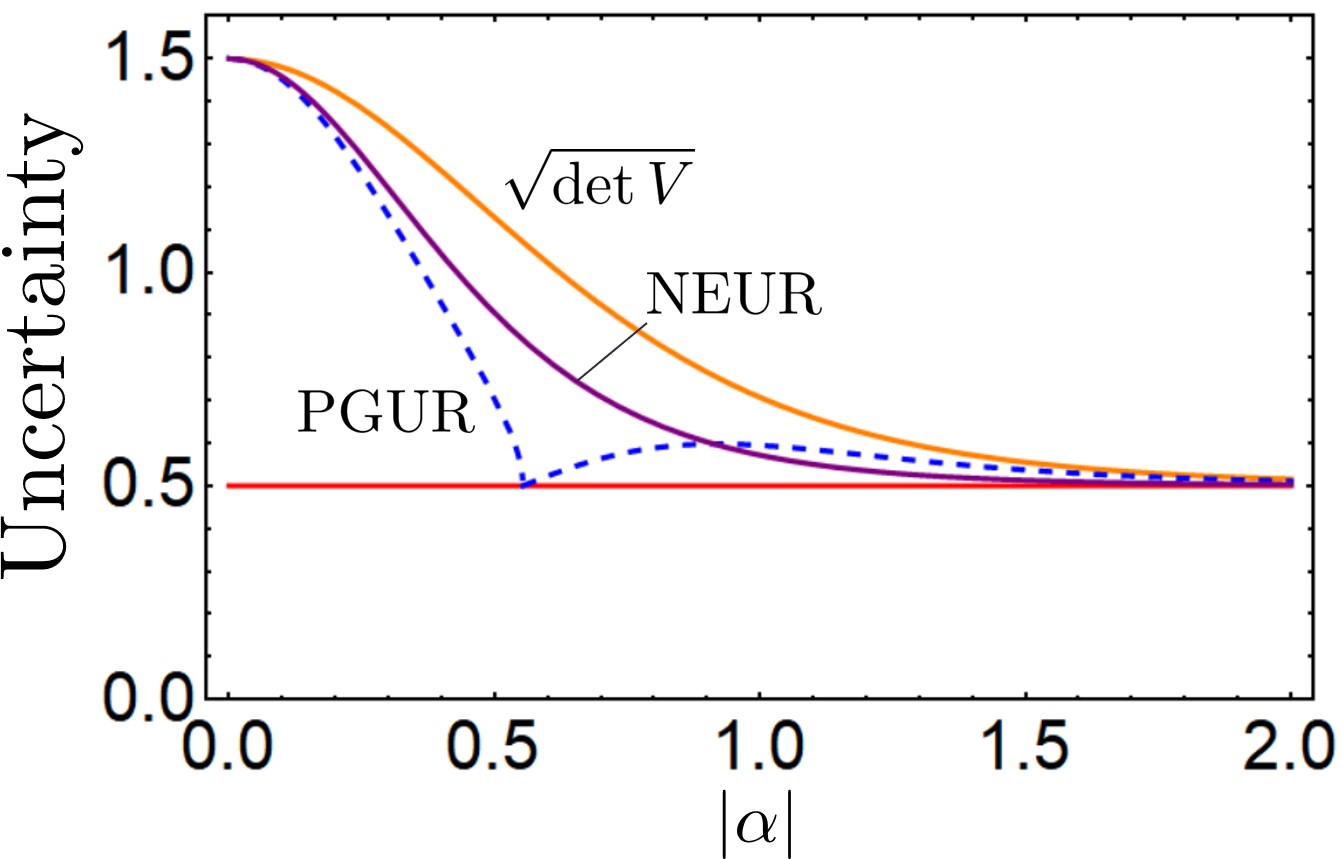}
	\caption{Plot of $\sqrt{\det  V}$  and the bounds of the NE, PG and RS uncertainty relations for photon-added coherent states $|\psi_{pacs}\rangle$ against $|\alpha|$.}\label{FigPACS}
\end{figure}

\subsection{Generalization for two-mode system}

The NE uncertainty relation can be straightforwardly generalized to the case of two-mode system in the same manner as with the case of single mode system.

\begin{theorem}
	For two-mode system $\hat \rho^{AB}$ with the non-Gaussianity measure $\mathcal N (\rho^{AB})$ and the entropy $S(\rho^{AB})$, we have
	\begin{align}\label{NEUR3}
	h(\nu_+)+h(\nu_-)\geq S(\rho^{AB})+N(\rho^{AB}),
	\end{align}
	where $\nu_{+}\geq \nu_{-}$ are symplectic eigenvalues of covariance matrix for canonical operators of two-mode system.
\end{theorem}
\begin{proof}	
	Covariance matrix for two-mode system is a real $4\times 4$ symmetric positive block matrix 
	\begin{align}
	V=\begin{pmatrix}
	A & C \\
	C^T & B
	\end{pmatrix},
	\end{align}
	where $A$, $B$ and $C$ are $2\times 2$ real matrices. Its symplectic eigenvalues are determined by two symplectic invariants $\det V$ and $\Delta V=\det A+ \det B +2\det C$ as
	\begin{align}\label{SD}
	\nu_\pm=\sqrt{\frac{\Delta V\pm\sqrt{(\Delta V)^2-4 \det V}}{2}}.
	\end{align}
	Thus, we have entropy of two-mode Gaussian state in the form of 
	\begin{align}
	S(\rho^{AB}_G)=h(\nu_+)+h(\nu_-),
	\end{align}
	since one can always achieve symplectic diagonalization leading to thermal states without affecting the entropy. Hence, with the monotonicity of quantum R\'enyi relative entropy in $\alpha$
	\begin{align}
	S(\rho^{AB}\|\rho_G^{AB})&=S(\rho_G^{AB})-S(\rho^{AB})\\ &\geq S_{1/2}(\rho^{AB}\|\rho_G^{AB})=N(\rho^{AB}),\nonumber
	\end{align}
	we have desired inequality \eqref{NEUR3} by adding $S(\rho^{AB})$ on both sides,
	\begin{align}
	S(\rho_G^{AB})\geq S(\rho^{AB})+N(\rho^{AB}),
	\end{align}
	where $\rho_G^{AB}$ is the reference Gaussian state with the same covariance matrix of $\rho^{AB}$.
\end{proof}

The NE uncertainty relation \eqref{NEUR3} has desired properties. First, all quantities in \eqref{NEUR3} are invariant under Gaussian unitary transformations, thus it is possible to verify the inequality regardless of principle axes. Second, the relation \eqref{NEUR3} is saturated for Gaussian states, and imposes enhanced restrictions on $\nu_\pm$ for non-Gaussian states. Generalizaed RS uncertainty relations for multimode system \cite{Simon1994} indicate symplectic eigenvalues should be larger than 1/2, that is $\nu_\pm\geq 1/2$. This is equivalent to the trivial case of \eqref{NEUR3}, where both $S$ and $N$ vanish. It is because $h(x)$ gives us valid values only if $x\geq1/2$. Furthermore, for nontrivial cases, it may impose tighter restrictions on possible values of $\nu_\pm$. This behavior is illustrated in Fig. \ref{FigNEUR2} showing regions of possible values of $\nu_\pm$ with respect to overall bound $B=S(\rho^{AB})+N(\rho^{AB})$.

\begin{figure}[t]
	\centering
	\includegraphics[scale =0.4]{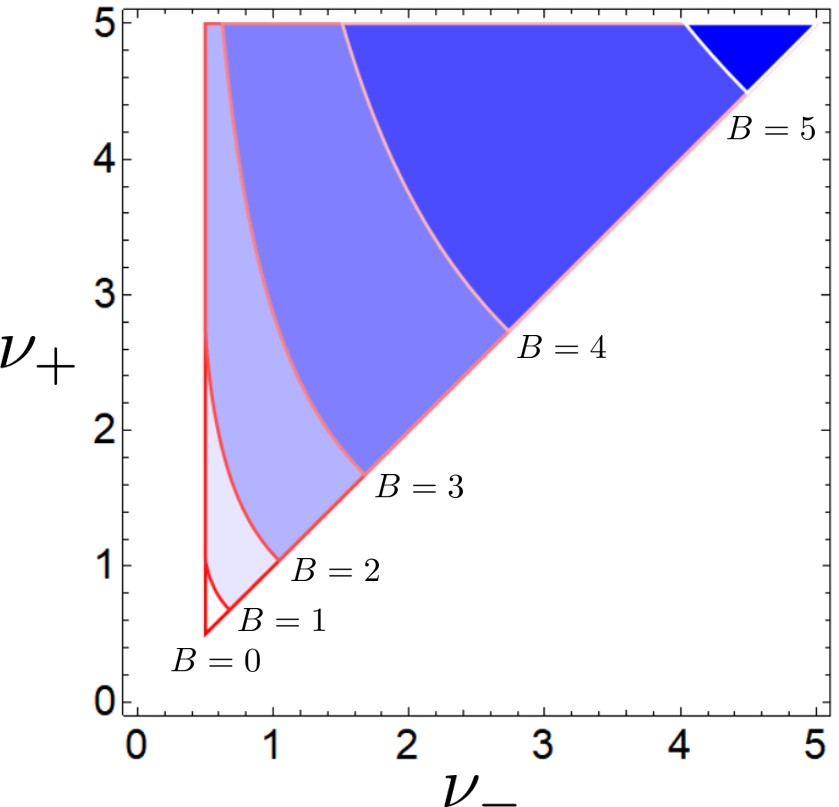}
	\caption{Enhanced restrictions on symplectic eigenvalues $\nu+\geq \nu_-$ according to the bound in inequality \eqref{NEUR3} denoted by $B=S(\rho^{AB})+N(\rho^{AB})$. With increase of $B$ from 0 to 5 distinguished by the contrast of blue, available regions of $\nu_{\pm}$ become more confined. }\label{FigNEUR2}
\end{figure}

\section{Application to entanglement detection of non-Gaussian states}

Entanglement has played crucial roles in quantum information science. However, even if a bipartite state is fully known, it is a NP-hard problem to verify whether it is entangled or not \cite{Gurvits2003}.  Entanglement criterion has been derived in \cite{Peres1996} by observing negativity of partial transposed states, since any separable state remains as a physical state, i.e. positive-semidefinite operator under partial transposition. 
For continuous variable systems, it has been generalized by Simon \cite{Simon2000} and Duan \cite{Duan2000} by addressing its physicality in view of uncertainty relations as follows. Partial transposition on the second mode of $\hat \rho^{AB}$ corresponds to the mirror reflection, $\hat p_B \rightarrow -\hat p_B$ in phase space. Accordingly, the covariance matrix is changed into 
\begin{align}
\tilde{V} = P_B V P_B = \begin{pmatrix}
A & \tilde{C} \\
\tilde{C}^T & \tilde{B}
\end{pmatrix},
\end{align}
with $P_B=diag(1,1,1,-1)$. If $\hat \rho^{AB}$ is a separable state, its partial transposed state should satisfy uncertainty relations, $\tilde{\nu}_{\pm}\geq 1/2$, where $\tilde{\nu}_\pm$ are the symplectic eigenvalues of $\tilde{V}$. Thus, according to the Simon-Duan criterion, the violation of uncertainty relation, i.e $\tilde{\nu}_{-}\ngeq 1/2$, implies entanglement of bipartite states. It was shown that Simon-Duan criterion is a necessary and sufficient condition for entanglement of Gaussian states. However, for non-Gaussian states, the satisfaction of inequality does not necessarily guarantee the state is separable.

To deal with the entanglement detection of non-Gaussian states, improved entanglement criteria were suggested in \cite{Nha2008,Walborn2009,Hertz2016,Nha2006,Nha2007}. 
For the same purpose, we suggest an entanglement criterion based on our NE uncertainty relation.  
According to Williamson theorem \cite{Williamson1936}, one can always find appropriate symplectic transformation diagonalizing it such that
\begin{align}
\tilde{V} \rightarrow \tilde{V}_d = S \tilde{V} S^T= \begin{pmatrix}
\tilde{\nu}_+ I & 0 \\
0 & \tilde{\nu}_- I
\end{pmatrix},
\end{align}
where $I$ is 2$\times$2 identity matrix. Thus, the violation of uncertainty relation, $\tilde\nu_-\ngeq 1/2$, is equivalent to observing $\hat{\sigma}^B=\text{Tr}_A[\hat U_S (\hat\rho^{AB})^{T_B} \hat U_S^\dagger]$ violates the RS uncertainty relation as pointed out in \cite{Hertz2016}, where $\hat U_S$ is the unitary operator corresponding to $S$. 
By applying this method to our NE uncertainty relation for single mode, one can obtain improved entanglement criterion for non-Gaussian states,
\begin{align}
h(\tilde{\nu}_-)\ngeq S(\sigma^B)+\mathcal N (\sigma^B)\rightarrow \text{entangled}.
\end{align}
In general, it is challenging to solve eigenvalue problem and we can instead use weaker but readily computable NE uncertainty relation \eqref{NEUR2} as 
\begin{align}\label{EntCri}
h(\tilde{\nu}_-)\ngeq -\ln \mu+\mathcal N_g (\sigma^B)\rightarrow \text{entangled}.
\end{align}
It is worth noting that one may encounter unphysical purity, i.e. $\mu>1$, in the process of verifying the inequality under partial transposition. In this case, one can conclude that the state is entangled. $\mu>1$ immediately indicates $\hat \sigma_B$ is unphysical.

\begin{figure}[t]
	\centering
	\includegraphics[scale =0.32]{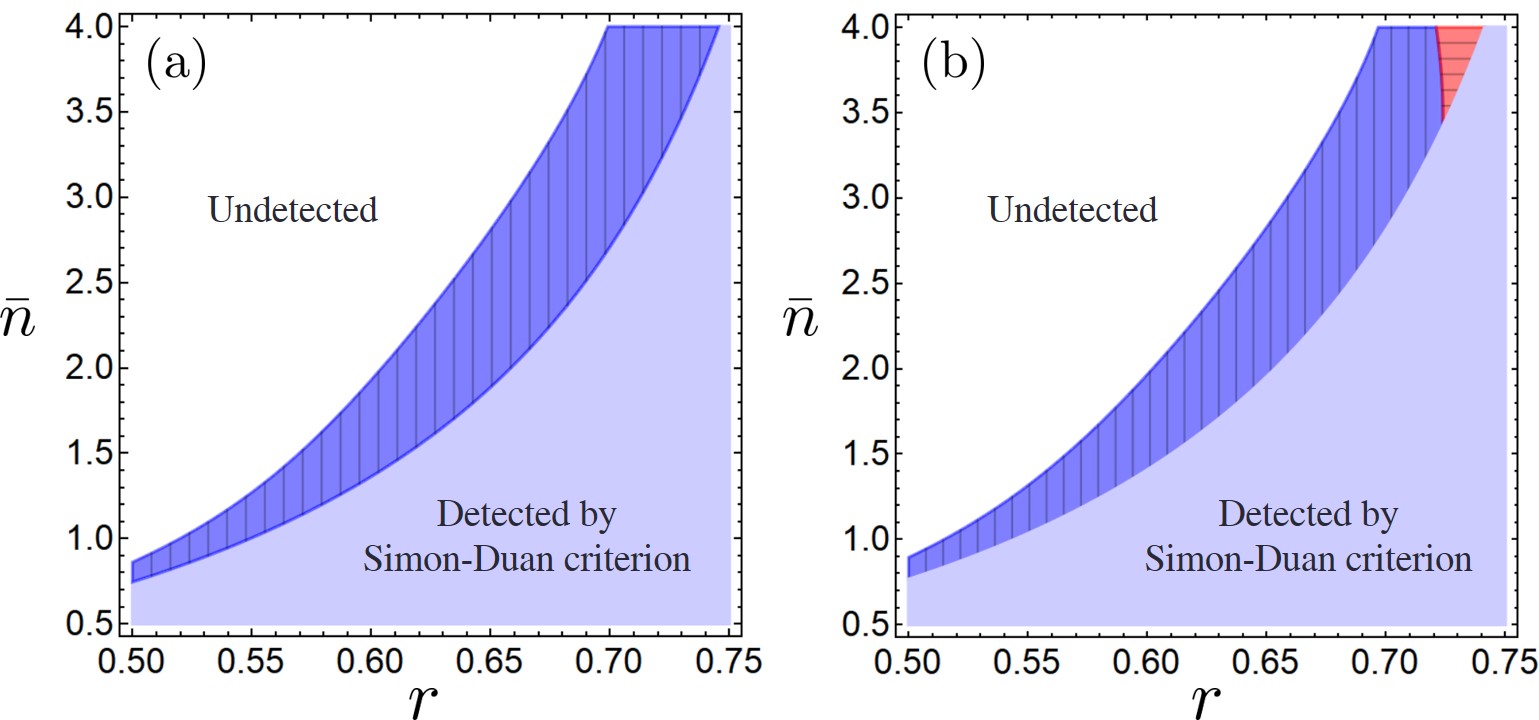}
	\caption{Plots of entanglement conditions for (a) $\alpha=0$ and (b) $\alpha=1$ detected by the Simon-Duan (Light blue) and the non-Gaussianity based criteria (Dashed blue and red) \eqref{EntCri} against the squeezing parameter $r$ and mean photon number $\bar n$ of a thermal state $\hat \tau^B(\bar{n})$.}\label{FigTMS}
\end{figure}

As an example, let us consider odd cat and thermal states coupled via two-mode squeezing operation, 
\begin{align}
\hat S^{AB}(\xi)\big( |\psi_-\rangle^A\langle\psi_-|\otimes \hat \tau^B(\bar{n}) \big) \hat S^{AB}(\xi)^\dagger,
\end{align}
where $\hat S^{AB}(\xi)=\exp{(\xi\hat a^\dagger\hat b^\dagger-\xi^* \hat a \hat b)}$ with the complex coupling $\xi=r e^{i\phi}$. For simplicity, we assume that they are squeezed along the direction of $\alpha$. In that case, we can put $\alpha=|\alpha|$ and $\xi=r$ without loss of generality. Here, $\alpha$ and $r$ determines the degree of non-Gaussianity and entanglement, respectively, while mean photon number $\bar n$ gives mixedness of overall states by adding thermal noise on it.

We show entanglement conditions with respect to $r$ and $\bar n$  detected via the non-Gaussianity based entanglement criterion \eqref{EntCri} for $\alpha=0,1$ in Fig. \ref{FigTMS}. 
In both cases, graphs show that our entanglement criterion (dashed blue and red) discovers undetected region by Simon-Duan criterion (light blue).
Additionally, we note the region denoted by dashed red on the right corner of Fig. \ref{FigTMS}(b) indicates entanglement discovered by observing unphysical purity larger than unity.

\section{Conclusion}

With the inequality \eqref{NEUR} we have provided an improved uncertainty relation by taking into account the degree of non-Gaussianity and mixedness that are quantified based on fidelity and von Neumann entropy, respectively. We have shown that this inequality, so-called NE uncertainty relation, includes RS relation with invariance under linear canonical transformations, and further it is saturated by all Gaussian and number states. To avoid challenging eigenvalue problems, we have also presented weaker but readily computable inequality \eqref{NEUR2} by using attainable quantities in phase-space description. We have pointed out that our uncertainty relations provide a strictly stronger bound for a non-Gaussian state even using the weaker version of inequality.  We have generalized the NE uncertainty relations to the case of a two-mode system and exhibited its enhanced restrictions on symplectic eigenvalues according to the increase of non-Gaussanity and mixedness.

As an application of our uncertainty relation, we have considered entanglement detection of non-Gaussian states. Due to the property that it gives strictly stronger bound for non-Gaussian states, we have obtained an improved version of Simon-Duan criterion \eqref{EntCri} for non-Gaussian states. To examine how it works, we have considered odd cat states coupled to thermal states by a two-mode squeezing operation. As a result, we have seen that ours can discover entangled states undetected by the Simon-Duan criteron.

\begin{acknowledgments}
We thank Prof. Son for discussions on the topic. This work is supported by the NPRP Grant No. 8-352-1-074 from the Qatar National Research Fund.
\end{acknowledgments}

\bibliography{mybibfile}

\end{document}